\documentclass[journal]{IEEEtran}
\usepackage{charter}
\usepackage{mathrsfs}
\usepackage{mathbbol}
\usepackage{latexsym}
\usepackage{pifont}
\usepackage{epsfig}
\usepackage{array}
\usepackage{url}
\usepackage{amssymb}
\usepackage{amsmath}
\usepackage{amsfonts}
\usepackage{subfig}
\usepackage{statex}
\usepackage{graphicx}
\usepackage{float}
\usepackage{hyperref}
\usepackage{verbatim}
\usepackage{ftnright}

\def\Pr{{\mathbb P}}
\long\def\symbolfootnote[#1]#2{\begingroup
\def\thefootnote{\fnsymbol{footnote}}\footnote[#1]{#2}\endgroup}

\begin{document}

\newtheorem{theorem}{Theorem}
\newtheorem{lemma}{Lemma}
\newtheorem{proposition}{Proposition}
\newtheorem{corollary}{Corollary}[theorem]
\newtheorem{definition}{Definition}
\newtheorem{remark}{Remark}[theorem]

\title{Throughput of Rateless Codes over Broadcast Erasure Channels}


\author{\IEEEauthorblockN{Yang Yang, and Ness B. Shroff,~\IEEEmembership{Fellow,~IEEE}}
}


\maketitle
\begin{abstract}
In this paper, we characterize the throughput of a broadcast network with $n$ receivers using rateless codes with block size $K$. We assume that the underlying channel is a Markov modulated erasure channel that is \emph{i.i.d.}~across users, but can be  correlated in time. We characterize the system throughput asymptotically in $n$. Specifically, we explicitly show how the throughput behaves for different values of the coding block size $K$ as a function of $n$, as $n\to\infty$. {\color{black}For finite values of $K$ and $n$, under} the more restrictive assumption of {\color{black}Gilbert-Elliott} channels, we are able to provide a lower bound on the maximum achievable throughput. Using simulations we show the tightness of the bound with respect to system parameters $n$ and $K$, and find that its performance is significantly better than the previously known lower bounds. 
\end{abstract}




\section{Introduction}
In this work\footnote{The preliminary version of this paper has appeared in \cite{YangShroffMobihoc2012}.
}, we study the throughput of a wireless broadcast network with $n$ receivers using rateless codes. In this broadcast network, channels between the transmitter and the receivers are modeled as packet erasure channels where transmitted packets may either be erased or successfully received. This model describes a situation where packets may get lost or are not decodable at the receiver due to a variety of factors such as channel fading, interference or checksum errors. We assume that the underlying channel is a Markov modulated packet erasure channel that is \emph{i.i.d.}~across users, but can be correlated in time. We let $\gamma$ denote the steady state probability that a packet is transmitted successfully on the erasure channel.

Instead of transmitting the broadcast data packet one after another through feedback and retransmissions, we investigate a class of coding schemes called rateless codes (or fountain codes). In this coding scheme, $K$ broadcast packets are encoded together prior to transmission. $K$ is called the coding block size. A rateless encoder views these $K$ packets as $K$ input symbols and can generate an arbitrary number of output symbols (which we call coded packets) as needed until the coding block is decoded. Although some coded packets may get lost during the transmission, rateless decoder can guarantee that any $K(1+\varepsilon)$ coded packets can recover the original $K$ packets {\color{black} with high probability}, where $\varepsilon$ is a positive number that can be made arbitrarily small at the cost of coding complexity.
Examples of rateless erasure codes include Raptor codes \cite{TON}, LT Codes \cite{LT-code} and random linear network codes \cite{Ho}, where the former two are more efficient when $K$ is very large and random linear network code is more efficient when $K$ is relatively small and the field size of packets is large. The {\color{black}best} encoding and decoding complexity of rateless codes {\color{black}(e.g. Raptor codes)} increase linearly as the coding block size $K$ increases. Further, increasing the coding block size can result in large delays and large receiver buffer size. Therefore, real systems always have an upper bound on the value of $K$.

We consider broadcast traffic and a discrete time queueing model, where the numbers of packet arrivals over different time slots are independent and identically distributed and the packet length is a fixed value. 
We let $\lambda$ denote the packet arrival rate and assume that the encoder waits until there are at least $K$ packets in the queue and then encodes the first $K$ of them as a single coding block. In this case, the largest arrival rate that can be stabilized is equal to the average number of packets that can be transmitted per slot, which we call the throughput.
Therefore, we only need to characterize the throughput that can be achieved using rateless codes under parameters $K$ and $n$. As described in Figure~\ref{ArrivalProcessFig}, the channel dynamics for the $i^{\text{th}}$ receiver is denoted by a stochastic process $\{X_{ij}\}_{j\in\mathbb{N}}$, where $j$ is the index of the time slot in which one packet can be transmitted and $X_{ij}$ is the channel state of $i^{\text{th}}$ receiver during the transmission of the $j^{\text{th}}$ packet. We capture a fairly general correlation structure by letting the current channel state be impacted by the channel states in previous $l$ time slots, where $l$ can be any number. As the number of receivers $n$ approaches infinity, we show that the throughput is nonzero only if the coding block size $K$ increases at least as fast as $\log n$. In other words, if $c\triangleq\lim_{n\to\infty}\frac{K}{\log n}$, the asymptotic\footnote{the asymptotic is with respect to increasing the number of receivers~$n$} throughput is positive whenever $c>0$. In Theorem~\ref{Theorem1}, by utilizing large deviation techniques, we give an explicit expression for the asymptotic throughput, which is a function of $K$, $n$, $\gamma$ and the channel correlation structure.

\begin{figure}[h]
\begin{center}
\includegraphics [scale=0.55]{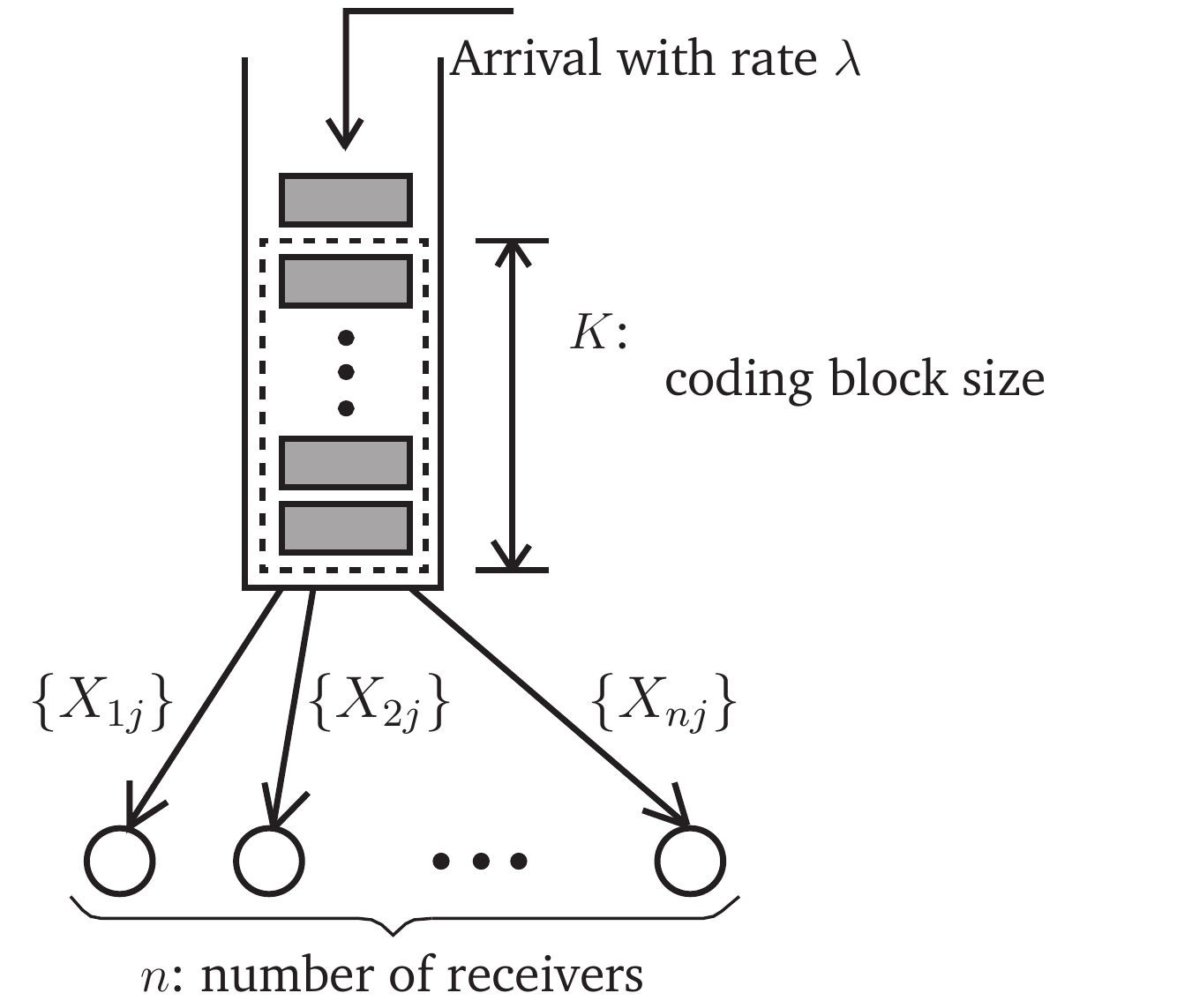}
\caption{Broadcast with discrete time queueing model}
\label{ArrivalProcessFig}
\end{center}
\end{figure}
{\color{black}To study the non-asymptotic behavior of the system, we make a more restrictive channel assumption that 
the current channel state is impacted by only the channel state in previous 1 time slot, which is the so called Gilbert-Elliott channel model.
In this case, for any finite $K$ and $n$, we find a lower bound on the throughput in terms of the transmission time of a system with larger $K$ and $n$. As a special case when the channels are memoryless, if $\frac{K}{\log n}$ is kept constant, this lower bound reveals that the throughput will follow a decreasing pattern as the number of receivers $n$ increases. By combining this result with the characterization of the asymptotic throughput, we are able to provide a lower bound on the maximum achievable throughput for any finite values of $K$ and $n$. This lower bound captures the asymptotic throughput in the sense that when $n$ approaches infinity, it coincides with the asymptotic throughput.
}
\subsection{Related Work}
Among the works that investigate the throughput over erasure channels, \cite{Cogill}, \cite{CogillBrooke}, \cite{Swapna} and \cite{Online} are the most relevant to this work. 
In \cite{Swapna}, the authors investigate the asymptotic throughput as a function of $n$ and $K$ and also show that the asymptotic throughput will be non-zero only if K at least scales with $\log n$. However, they only consider the channel correlation model with $l=1$ and use a completely different proof technique. Moreover, no explicit expression on the asymptotic throughput is provided. In \cite{Cogill} and \cite{CogillBrooke}, two lower bounds on the maximum achievable rate $\lambda$ are provided. However, their bound does not converge to the asymptotic throughput when $n$ approaches infinity. 
Moreover, our bound is shown to be better in a variety of simulation settings with finite $K$ and $n$, as will be showed in Section~\ref{Simulation}. {\color{black}In \cite{Online}, the authors consider the case when instantaneous feedback is provided from every user after the transmission of each decoded packets, while we only assume that feedback is provided after the entire coding block has been decoded.}

\subsection{Key Contributions}
The main contributions of this work are summarized as follows:
\begin{itemize}
\item We give an explicit expression {\color{black}for} the asymptotic throughput of the system when the number of receivers $n$ approaches infinity {\color{black}for any values of $K$ as a function of $n$} under the erasure channel with any levels of correlation. (Theorem~\ref{Theorem1})
\item {\color{black} Under the Gilbert-Elliott channel model ($l=1$), for any finite $K$ and $n$, we find a lower bound on the throughput in terms of the transmission time of a system with larger $K$ and $n$. 
As a special case, when channels are memoryless ($l=0$), this lower bound reveals that when $K$ grows with $n$ in a way that the ratio $\frac{K}{\log n}$ is kept constant, the throughput follows a decreasing pattern as $n$ increases.
(Theorem~\ref{MonotoneTh})}
{\color{black}\item We provide an asymptotically tight lower bound on the maximum achievable throughput for any values of $K$ and $n$ under the Gilbert-Elliott channel model ($l=1$) and show that its performance is significantly better than the previously known bounds in \cite{Cogill} and \cite{CogillBrooke}. (Theorem~\ref{LowerBoundTheorem})}
\end{itemize}

The rest of this paper is organized as follows.
In Section~\ref{sysModel} we describe our model and assumptions. In Section~\ref{AsymptoticSection} we give the characterization of the asymptotic throughput. In Section~\ref{RateRegionSection} we provide a lower bound on the maximum achievable throughput for any finite values of $K$ and $n$. In Section~\ref{Simulation} we use simulations to verify our theoretical results. Detailed proofs on all the theorems can be found in Section~\ref{Proofs}. Finally, in Section~\ref{Conclusion}, we conclude the paper.

\section{System Model}\label{sysModel}

We consider a broadcast channel with $n$ receivers. Time is slotted, and the numbers of broadcast packet arrivals over different time slots are i.i.d.~with finite variance. We denote the expected number of packet arrivals per slot as the packet arrival rate $\lambda$. The transmission starts when there are more than $K$ packets waiting in the incoming queue intended for all the receivers. Instead of transmitting these packets one after another using feedback and retransmissions, we view each data packet as a symbol and encode the first $K$ of them into an arbitrary number of coded symbols as needed using rateless code (For example, Raptor Code \cite{TON} or random linear network code \cite{Ho}) until the coding block is decoded. These $K$ packets together form a single coding block with $K$ being called block size. During the transmission, the coded symbols are transmitted one after another. 

Each receiver sends an ACK feedback signal after it has successfully decoded the $K$ packets.
In the following context, the term {\it packet} and {\it symbol} are used interchangeably.

We model the broadcast channel as a slotted broadcast packet erasure channel where one packet can be transmitted per slot.
The channel dynamics can be represented by a stochastic process $\{X_{ij}\}_{1\leq i\leq n,j\in\mathbb{N}}$, where $X_{ij}$ is the state of channel between transmitter and the $i^{\text{th}}$ receiver during the transmission of $j^{\text{th}}$ packet (we also call it the $j^{\text{th}}$ time slot in the $i^{\text{th}}$ channel), which is given by
	\begin{align}
		X_{ij}=\left\{
			\begin{array}{ll}
			1&\text{$j^\text{th}$ packet in the $i^\text{th}$ channel is}\\
			&\text{ successfully received}\\
			0&\text{otherwise}
			\end{array}
		\right.\notag.
	\end{align}
	We assume that the dynamics of the channels for different receivers are independent and identical. More precisely, for all $1\leq i\leq n$, $\{X_{ij}\}_{j\geq1}$ are independent and identical processes.
	
	Since, in practice, the channel dynamics are often temporarily correlated, we investigate the situation where the current channel state distribution depends on the channel states in the preceding $l$ time slots. More specifically, for $\mathcal{F}_{im}=\{X_{ij}\}_{j\leq m}$ and fixed $l$, we define $\mathcal{H}_{im}=\{X_{im},$ $\ldots,$ $X_{i(m-l+1)}\}$ for $m\geq l\geq 1$ with $\mathcal{H}_{im}=\{\varnothing,\Omega\}$ for $l=0$, and assume that $\Pr[X_{i(m+1)}=1|\mathcal{F}_{im}]=\Pr[X_{i(m+1)}=1|\mathcal{H}_{im}]$ for all $m\geq l$. To put it another way, when $l\geq1$, the state $\left(X_{im},\ldots,X_{i(m-l+1)}\right)$, $m\geq l$ forms a Markov chain. Denote by $\Pi$ the transition matrix of the Markov chain $\left\{\left(X_{im},\ldots,X_{i(m-l+1)}\right)\right\}_{m\geq l}$, where
	\begin{align}
	\Pi=[\pi(s,u)]_{s,u\in\{0,1\}^l},\notag
	\end{align}
with $\pi(s,u)$ being the one-step transition probability from state $s$ to state $u$. Throughout this paper, we assume that $\Pi$ is irreducible and aperiodic, which ensures that this Markov chain is ergodic \cite{MarkovChain}. Therefore, for any initial value $\mathcal{H}_{l}$, the parameter $\gamma_i$ is well defined and given by
\begin{align}
	\gamma_i=\lim_{m\to\infty}\Pr[X_{im}=1],\notag
\end{align}
and, from the ergodic theorem \cite{MarkovChain} we know
\begin{align}
	\Pr\left[\lim_{m\to\infty}\frac{\sum_{j=1}^m X_{ij}}{m}=\gamma_i\right]=1.\notag
\end{align}
Since $\{X_{ij}\}_{j\geq1}$ for all $1\leq i\leq n$ are i.i.d., we denote $\gamma=\gamma_i$, for all $1\leq i\leq n$.

Using near optimal rateless codes, such as Raptor Codes \cite{TON}, LT Codes \cite{LT-code} and random linear network codes \cite{Ho}, only slightly more than $K$ coded symbols are needed to decode the whole coding block. For simplicity, here we assume that any combination of $K$ coded symbols can lead to a successful decoding of the $K$ packets.

According to the above system model, we have the following definitions:
\begin{definition}\label{TiK}{\rm The number of time slots (number of transmitted coded symbols) needed for user $i$ to successfully decode $K$ packets is defined as
    \begin{align}
      T_i(K)=\min{\Bigg\{}m{\bigg|}\sum_{j=1}^m X_{ij}\geq K{\Bigg\}}.\notag
    \end{align}
}
\end{definition}
\begin{definition}\label{TnK}{\rm
The number of time slots (number of transmitted coded symbols) needed to complete the transmission of a single coding block to all the receivers is defined as
    \begin{align}
      T(n,K)=\max\left\{T_i(K),i=1,2,\ldots,n\right\}.\notag
    \end{align}
}
\end{definition}
\begin{definition}[Initial State]
Since the current channel state depends on the channel states in the previous $l$ time-slots, for each receiver $i$, by assuming that the system starts at time slot $1$, we define the initial state of receiver $i$ as
{\color{black}
\begin{align}
\mathcal{E}_i=\left[X_{i(-l+1)},X_{i(-l+2)},\ldots,X_{i0}\right]\in\{0,1\}^l.\notag
\end{align}
}The initial state for all the receivers is then denoted as 
{\color{black}$\mathcal{E}\triangleq\left[\mathcal{E}_1,\mathcal{E}_2,\ldots,\mathcal{E}_n\right]$}.
\end{definition}
\begin{definition}[Throughput]\label{eta} 
For a system with an infinite backlog of packets, we define throughput $\eta(n,K)$ as the long term average number of packets that can be transmitted per slot. More precisely,
{\color{black}
\begin{align}
\eta(n,K)=K\times\lim_{t\to\infty}\frac{R(t)}{t},\notag
\end{align}
where $R(t)$ is the number of successfully transmitted coding blocks in $t$ time slots. For any finite values of $K$ and $n$, it is easy to check that 
$\{\mathcal{E}^h,T^h(n,K)\}_h$ is a finite-state ergodic Markov renewal process, where $\mathcal{E}^h$ and $T^h(n,K)$ denote the initial state and the transmission time of the $h^{\text{th}}$ coding block, respectively. Then, $R(t)$ is the total number of state transitions that occur in $t$ time slots of this Markov renewal process. Therefore, from \cite{RenewalTheory} we know that

\begin{align}
\eta(n,K)\overset{\text{a.s.}}{=}\frac{K}{\mathbb{E}[T(n,K)]}=\frac{K}{\mathbb{E}[\mathbb{E}[T(n,K)|\mathcal{E}]]},\label{Thput}
\end{align}
where the outer expectation in the last term denote the expectation with respect to the steady state distribution of the embedded Markov chain $\{\mathcal{E}^h\}_h$.
}
\end{definition}

\section{Asymptotic Throughput}\label{AsymptoticSection}
Before presenting the main results, we need to introduce some necessary definitions. First, define a mapping $f$ from the state space of the Markov chain $\{0,1\}^l$
to $\{0,1\}$ as
\begin{align}
	f{\big(}(X_{im},\ldots,X_{i(m-l+1)}){\big)}=X_{im}\notag.
\end{align}
Then, given a real number $\theta$, we define a matrix $\Pi_\theta$ as
\begin{align}
\Pi_\theta=\left\{
	\begin{array}{ll}
	\left[\pi(s,u)e^{\theta f(u)}\right]_{ s,u\in\{0,1\}^l}&\text{when }l\geq1\\
	\left[\gamma e^{\theta}\right]&\text{when }l=0
	\end{array}
\right.\notag.
\end{align}
Last, define a standard large deviation rate function $\Lambda(\beta,\Pi)$ as
\begin{align}
	\Lambda(\beta,\Pi)=\sup_\theta\{\theta\beta-\log\rho(\Pi_\theta)\}\label{lambdadef},
\end{align}
where $\rho(\Pi_\theta)$ denotes the Perron-Frobenious eigenvalue of $\Pi_\theta$ (See Theorem 3.1.1 in \cite{Dembozeitouni}), which is the largest eigenvalue of $\Pi_\theta$.

The asymptotic throughput for any values of $K$ as a function of $n$ is characterized by the theorem below:
\begin{theorem}\label{Theorem1}
Assume that $K$ is a function of $n$ and the value of $\lim_{n\to\infty}\frac{K}{\log n}$ exists, which we denote by $c\triangleq\lim_{n\to\infty}\frac{K}{\log n}$, then 
we have
      \begin{align}
        \lim_{n\to\infty}\eta(n,K)=\sup\left\{\beta{\bigg|}c\geq\frac{\beta}{\Lambda(\beta,\Pi)},0\leq\beta<\gamma\right\}\label{god}.
      \end{align}
\end{theorem}
\begin{proof}
see Section~\ref{p1}.
\end{proof}
From Theorem~\ref{Theorem1}, we know that, if the coding block size $K$ is set to be a function of the network size $n$, then we can characterize the asymptotic throughput when $n$ approaches infinity in an explicit form. {\color{black}Equation~(\ref{god}) implies that the asymptotic throughput 
is a function of $\gamma$, $\lim_{n\to\infty}K/\log n$ and the channel correlation structure indicated by $\Pi$.}

By Theorem~\ref{Theorem1}, the asymptotic throughput in the special cases when $K\in o(\log n)$ and $K\in \omega(\log n)$ are given in the following corollary.
\begin{corollary}\label{corollary1}
Assume that $K$ is a function of $n$.
We then have
\begin{enumerate}
	\item if $K\in o(\log n)$, then\footnote{We use standard notations: $f(n)=o(g(n))$ \\if $\lim_{n\to\infty}\frac{f(n)}{g(n)}=0$ and $f(n)=\omega(g(n))$ if $\lim_{n\to\infty}\frac{f(n)}{g(n)}$ diverges}
  	\begin{align}
		\lim_{n\to\infty}\eta(n,K)=0.\notag
	\end{align}
  \item if $K\in\omega(\log n)$, then
  	\begin{align}
		\lim_{n\to\infty}\eta(n,K)=\gamma.\notag
	\end{align}
\end{enumerate}
\end{corollary}
\begin{proof}
1) If $K\in o(\log n)$, then $c=\lim_{n\to\infty}\frac{K}{\log n}=0$ and we have
	\begin{align}
		\left\{\beta{\bigg|}c\geq\frac{\beta}{\Lambda(\beta,\Pi)},0\leq\beta<\gamma\right\}=\left\{0\right\}.\notag
	\end{align}
	According to Theorem~\ref{Theorem1}, we get
	\begin{align}
	\lim_{n\to\infty}\eta(n,K)=\sup\{0\}=0.\notag
	\end{align}
2) If $K\in \omega(\log n)$, then $c=\lim_{n\to\infty}\frac{K}{\log n}=\infty$ and we have
\begin{align}
		\left\{\beta{\bigg|}c\geq\frac{\beta}{\Lambda(\beta,\Pi)},0\leq\beta<\gamma\right\}=[0,\gamma).\notag
\end{align}
	According to Theorem~\ref{Theorem1}, we get
	\begin{align}
	\lim_{n\to\infty}\eta(n,K)=\sup[0,\gamma)=\gamma.\notag
	\end{align}
\end{proof}
Corollary~\ref{corollary1} says that the throughput will vanish to $0$ {\color{black}as $n$ becomes large,} when $K$ does not scale as fast as $\log n$. Whereas when $K$ scales faster than $\log n$ (Or more specifically, when $K\in \omega(\log n)$), {\color{black}throughput approaches the capacity $\gamma$ of the system in the limit.} It should be noted that Theorem~\ref{Theorem1}, together with Corollary~\ref{corollary1}, are a generalized version of Theorem 1 in \cite{Swapna}, which only consider the case when $l=1$ and does not give an explicit expression for the asymptotic throughput.

As a special case when the channels are memoryless ($l=0$), we can express $\Lambda(\beta,\Pi)$ in a closed form, as shown in the corollary below.
\begin{corollary}\label{finally}
Assume that $K$ is function of $n$ and the channels are memoryless ($l=0$), we have\\ if $\lim_{n\to\infty}\frac{K}{\log n}=c$, where $c$ is a positive constant, then
\begin{align}
&\lim_{n\to\infty}\eta(n,K)=\notag\\
&\sup\left\{\beta{\bigg|}\log\frac{\beta}{\gamma}+\frac{1-\beta}{\beta}\log\frac{1-\beta}{1-\gamma}\geq\frac{1}{c},0\leq\beta<\gamma\right\}.
\end{align}
\end{corollary}
\begin{proof} When $l=0$, $\Pi_\theta=[\gamma e^{\theta}]$ is a degenerate matrix with a single entry and $\rho(\Pi_\theta)=\rho(\gamma e^{\theta})=\gamma e^{\theta}$. Therefore we have, according to Equation~(\ref{lambdadef}),
\begin{align}
\Lambda(\beta,\Pi)=\beta\log\frac{\beta}{\gamma}+(1-\beta)\log\frac{1-\beta}{1-\gamma}.\notag
\end{align}
\end{proof}

\section{{\color{black}Throughput lower bound for finite $K$ and $n$}}\label{RateRegionSection}

For all rateless coding schemes, the encoding and decoding complexity increases linearly in $K$, the size of the coding block. Moreover, the value of $K$ determines the receiver buffer size. Therefore, in reality, the value of $K$ is often limited by the decoder buffer size or the computational power of both sender and receiver. {\color{black}We then} have to consider the case when $K$ is finite and need to answer the following questions: For a given number of receivers $n$, channel statistics, and a maximum available coding block size $K$, what is maximum packet arrival rate $\lambda$ that can be supported by this system? For a specific number of receivers and channel statistics, if we are given a target packet arrival rate $\lambda$, how can we design the value of $K$ in the system such that the target arrival rate can be supported?

{\color{black}In order to answer these questions, we make a more restrictive channel assumption that the current channel state is impacted by only the channel state in the previous 1 time slot. In other words, we have the Gilbert-Elliott channel model.
Under this model, for any receiver $1\leq i\leq n$, the channel states $\{X_{ij}\}_{j\in\mathbb{N}}$ evolve with $j$ according to a two state Markov chain as illustrated in Figure~\ref{GilbertElliottIll}. Here, $p_{01}$ and $p_{10}$ are the transition probabilities between state $0$ and state $1$, and system capacity $\gamma=p_{01}/(p_{10}+p_{01})$. Based on this model, in the theorem below, we find a lower bound on the throughput for any $n$ and $K$ in terms of the expected transmission time of a system with larger $n$ and $K$.
\begin{figure}[h]
\begin{center}
\includegraphics [scale=0.9]{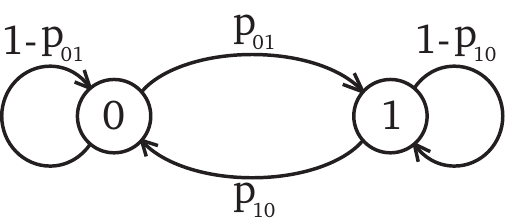}
\caption{Gilbert-Elliott Channel Model}
\label{GilbertElliottIll}
\end{center}
\end{figure}
\begin{theorem}\label{MonotoneTh}{\rm
Under the Gilbert-Elliott channel model ($l=1$) as illustrated in Figure~\ref{GilbertElliottIll}, for any $n\in\mathbb{N}$ $K\in\mathbb{N}$ and $\alpha\in\mathbb{N}$, we have\footnote{We denote an all-one vector with dimension $m$ as ${\bf1}_m$}
		\begin{align}
		\eta(n,K)>\frac{\alpha K}{\mathbb{E}\left[T(n^\alpha,\alpha(K+K_0))|\mathcal{E}={\bf1}_{n^\alpha}\right]}\notag,
		\end{align}
		where
		\begin{align}
		K_0=\min{\Bigg\{}m\geq0{\Bigg|}\sum_{d=0}^m(1-p_{10})^{d}p_{10}+p_{01}\geq1{\Bigg\}}\notag.
		\end{align}
}
\end{theorem}
}

{\color{black}
\begin{remark}
Observe that $K_0$ is independent of the choice of $n$ and $K$ and is only a function of channel dynamics. $K_0=0$ if and only if $p_{01}+p_{10}\geq1$.
\end{remark}
}
\begin{proof}
see Section~\ref{MonotoneProof}.
\end{proof}
{\color{black} When the channels are memoryless, the above theorem reduces to a simpler form, as shown in the following corollary.}
\begin{corollary}\label{MonotoneThMemoryless}{\rm
When the channels are memoryless ($l=0$), for any $n\in\mathbb{N},K\in\mathbb{N}$ and $\alpha\in\mathbb{N}$, we have 
\begin{align}
\eta(n,K)>\eta(n^\alpha,\alpha K)\notag.
\end{align}
}\end{corollary}

\begin{remark}
While Theorem~\ref{Theorem1} tells us that in order to achieve nonzero throughput, we can double the coding block size $K$ for every quadratic increase of $n$, which is to make $K/\log n$ a fixed value, it does not tell us anything about how the throughput will converge as $n$ approaches infinity. This corollary indicates that under the memoryless channel assumption, if we adapt the coding block size $K$ with the increase of network size $n$ in a way that $K/\log n$ is kept as a fixed value, then the throughput will follow a decreasing pattern before it reaches the asymptotic throughput.
\end{remark}

{\color{black}\begin{proof}
	The memoryless channels can be considered as a special case of the Gilbert-Elliott channels when $p_{10}=1-\gamma$, $p_{01}=\gamma$ and $K_0=0$, then by Theorem~\ref{MonotoneTh},
	\begin{align}
		\eta(n,K)&>\frac{\alpha K}{\mathbb{E}\left[T(n^\alpha,\alpha K)|\mathcal{E}={\bf1}_{n^\alpha}\right]}\notag\\
		&=\frac{\alpha K}{\mathbb{E}[T(n^\alpha,\alpha K)]}=\eta(n^\alpha,\alpha K)\notag.
	\end{align}
	The last equation follows by noting that $T(n^\alpha,\alpha K)$ is independent of initial state $\mathcal{E}$ when the channels are memoryless.
\end{proof}
}

By the help of the Theorem~\ref{Theorem1} and Theorem~\ref{MonotoneTh}, we can get a lower bound on the maximum stable throughput that can be achieved for any finite values of coding block size $K$ and network size $n$, as shown in the theorem below.
{\color{black}\begin{theorem}\label{LowerBoundTheorem}{\rm
For a broadcast network with $n$ receivers, and coding block size $K$, under the Gilbert-Elliott channel model ($l=1$), the throughput is lower bounded by
\begin{align}
\eta(n,K)>\frac{K}{K+K_0}\mathcal{R}\left(\frac{K+K_0}{\log n}\right),\notag
\end{align}
and thus the system with packet arrival rate $\lambda$ is stable if
\begin{align}
\lambda\leq\frac{K}{K+K_0}\mathcal{R}\left(\frac{K+K_0}{\log n}\right),\notag
\end{align}
where
\begin{align}
\mathcal{R}\left(r\right)=\sup\left\{\beta{\bigg|}r\geq\frac{\beta}{\Lambda(\beta,\Pi)},0\leq\beta<\gamma\right\},\text{ and}\notag\\
K_0=\min{\Bigg\{}m\geq0{\Bigg|}\sum_{d=0}^m(1-p_{10})^{d}p_{10}+p_{01}\geq1{\Bigg\}}\notag.
\end{align}
}\end{theorem}
\begin{remark}
As a special case, when the channels are memoryless, we have $p_{10}=1-\gamma$ and $p_{01}=\gamma$. It is easy to obtain that $K_0=0$. Therefore,
\begin{align}
&\eta(n,K)\notag\\
>&\sup\left\{\beta{\bigg|}\log\frac{\beta}{\gamma}+\frac{1-\beta}{\beta}\log\frac{1-\beta}{1-\gamma}\geq\frac{\log n}{K},0\leq\beta<\gamma\right\},\notag
\end{align}
and the system with packet arrival rate $\lambda$ is stable if
\begin{align}
\lambda\leq\sup\left\{\beta{\bigg|}\log\frac{\beta}{\gamma}+\frac{1-\beta}{\beta}\log\frac{1-\beta}{1-\gamma}\geq\frac{\log n}{K},0\leq\beta<\gamma\right\}\notag.
\end{align}
\end{remark}
}
\begin{proof}
From Equation~(\ref{lemma1result}) in Lemma~\ref{LemmaTK} we can see that when $K$ and $n$ are finite, the transmission time of a coding block $T(n,K)$ is light-tail distributed, meaning that it has finite variance. Then according to \cite{Lyapunov}, using Lyapunov method we know that the queue will be stable if the traffic intensity of this queue, which is defined as the packet arrival rate $\lambda$ over the service rate, is less than $1$. Therefore, the queue will be stable if the arrival rate $\lambda$ satisfies
{\color{black}
\begin{align}
\lambda&<\sup\left\{\mu{\bigg|}\frac{\mu}{K/\mathbb{E}[T(n,K)]}<1\right\}\notag\\
	&=\eta(n,K).\label{end1}
\end{align}
By Theorem~\ref{MonotoneTh} we know that, for any integer values of $\alpha$
\begin{align}
\eta(n,K)>\frac{\alpha K}{\mathbb{E}\left[T(n^\alpha,\alpha(K+K_0))|\mathcal{E}={\bf1}_{n^\alpha}\right]},\notag
\end{align}
implying that
\begin{align}
&\eta(n,K)\notag\\
>&\lim_{\alpha\to\infty}\frac{K}{K+K_0}\frac{\alpha (K+K_0)}{\mathbb{E}\left[T(n^\alpha,\alpha(K+K_0))|\mathcal{E}={\bf1}_{n^\alpha}\right]}.\label{end2}
\end{align}
Since $\alpha (K+K_0)/\log n^\alpha=(K+K_0)/\log n$ for any value of $\alpha$, then by Equation~(\ref{independentofe}) in the proof of Theorem~\ref{Theorem1}, we get
\begin{align}
&\lim_{\alpha\to\infty}\frac{\alpha (K+K_0)}{\mathbb{E}\left[T(n^\alpha,\alpha(K+K_0))|\mathcal{E}={\bf1}_{n^\alpha}\right]}\notag\\
=&\sup\left\{\beta{\bigg|}\frac{K+K_0}{\log n}\geq\frac{\beta}{\Lambda(\beta,\Pi)},0\leq\beta<\gamma\right\}\notag\\
=&\mathcal{R}\left(\frac{K+K_0}{\log n}\right),\notag
\end{align}
}which, by combining Equation~(\ref{end1}) and Equation~(\ref{end2}), completes the proof.
\end{proof}

In order to compare this lower bound on the maximum achievable rate with the existing bounds given in \cite{Cogill} and \cite{CogillBrooke}, we restate Theorem 2 in \cite{Cogill} and Theorem 7 in \cite{CogillBrooke} as the following.
{\color{black}
\begin{theorem}[Theorem 2 in \cite{Cogill} and Theorem 7 in \cite{CogillBrooke}]\label{CogillShrader}
In a broadcast network with $n$ receivers, coding block length $K$ and packet arrive rate $\lambda$,
 
1) when the channels are memoryless ($l=0$) with erasure probability $1-\gamma$ and $K>16$, the system is stable if
\begin{align}
\lambda<\frac{(1-\gamma)K}{K+(\log n+0.78)\sqrt{K}+2.61}.\notag
\end{align}

2) For Gilbert-Elliott channels ($l=1$) with state transition probability $p_{10}$ and $p_{01}$, when $1-p_{10}-p_{01}\geq0$ and $K\geq 21\log n-4$, the system is stable if
\begin{align}
\lambda<\frac{p_{01}K}{K+2\sqrt{(0.78K+3.37)\log n}+2.61}.\notag
\end{align}
\end{theorem}}
For ease of notation let us denote the bounds given in Theorem~\ref{CogillShrader} as the {\it CSE bound 1} and {\it CSE bound 2} respectively using the initials of the authors' last name.

{\color{black}
Firstly we should note that the CSE bounds are only valid when $K$ and $n$ satisfy certain conditions, while our bound is valid for any finite values of $K$ and $n$.
Secondly, our bound converges to the asymptotic throughput in the sense that as $n$ approaches infinity while keeping $K/\log n$ as a constant $c$, our bound on the maximum achievable rate will converge to the asymptotic throughput with parameter $c$. Or more specifically,}
{\color{black}
\begin{align}
&\lim_{n\to\infty}\frac{K}{K+K_0}\mathcal{R}\left(\frac{K+K_0}{\log n}\right)\notag\\
=&\mathcal{R}\left(\lim_{n\to\infty}\frac{K+K_0}{\log n}\right)=\mathcal{R}\left(c\right)=\lim_{n\to\infty}\eta(n,K),\label{AsympTight}
\end{align}
}which can be seen from Theorem~\ref{Theorem1} and Theorem~\ref{LowerBoundTheorem}.
 However, the CSE bounds are not asymptotically tight. When we keep the ratio $K/\log n$ to be a constant $c$, as $n$ or $K$ approaches infinity, CSE bound 1 even becomes trivial (approach $0$), which can be seen from the equation below.
\begin{align}
&\lim_{K\to\infty}\frac{(1-\gamma)K}{K+(\log n+0.78)\sqrt{K}+2.61}\notag\\
=&\lim_{K\to\infty}\frac{(1-\gamma)}{1+(1/c+0.78/K)\sqrt{K}+2.61/K}\notag\\
=&0\label{cse0}.
\end{align}
Next, in Section~\ref{Simulation}, we show that our bound outperforms the CSE bounds under various simulation settings.

\section{Simulation}\label{Simulation}
In this Section, we conduct simulation experiments to verify our main results. 
\subsection{Example 1}
\begin{figure}[h]
\begin{center}
\includegraphics [width=3.2in]{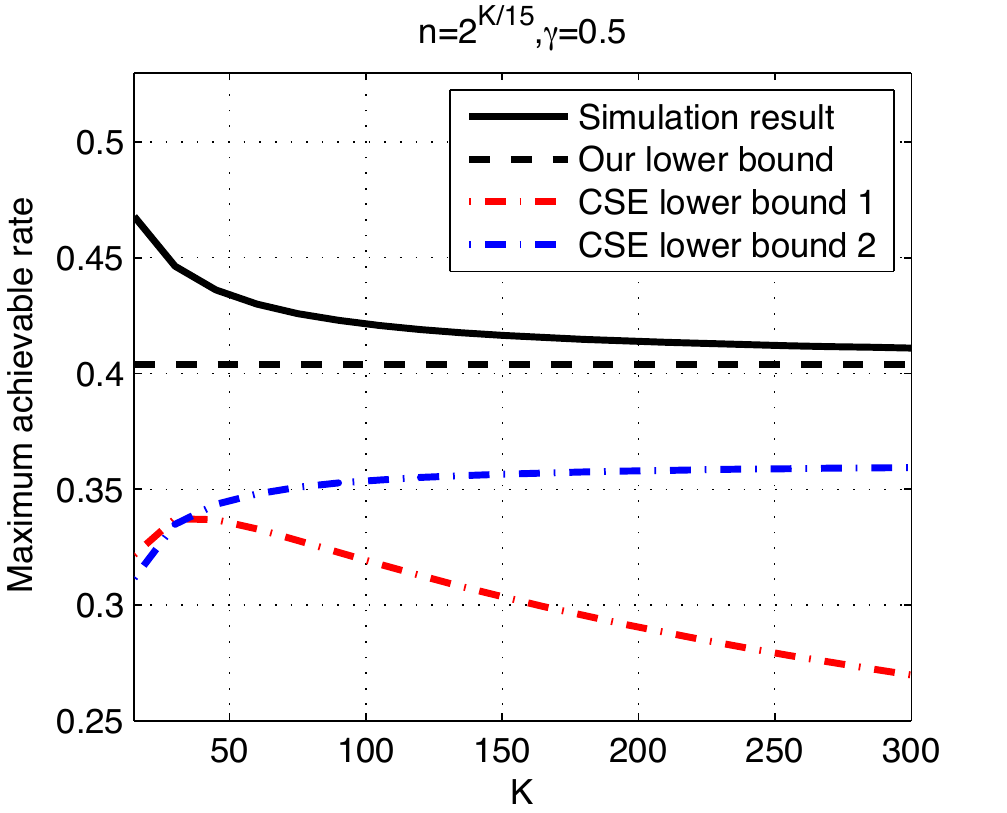}
\caption{Illustration for example 1}
\label{ex1fig1}
\end{center}
\end{figure}
\begin{figure}[h]
\begin{center}
\includegraphics [width=3.2in]{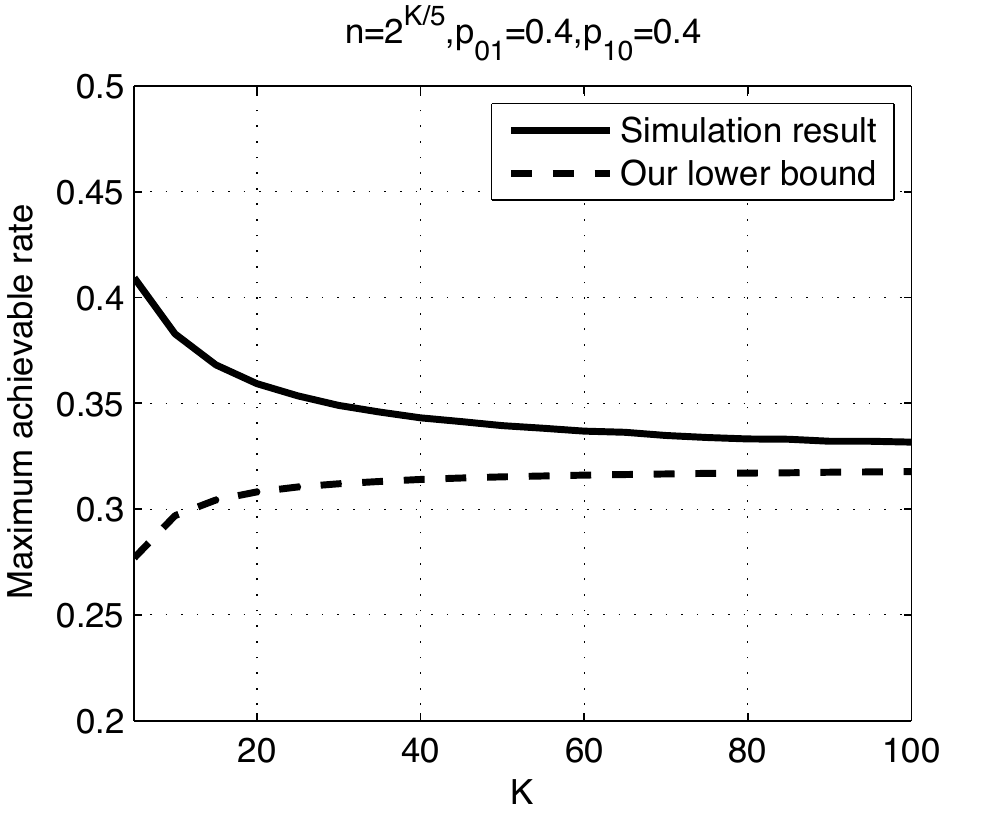}
\caption{Illustration for example 2}
\label{ex1fig2}
\end{center}
\end{figure}
This example verifies Theorem~\ref{Theorem1}, Corollary~\ref{MonotoneThMemoryless}, and Theorem~\ref{LowerBoundTheorem}. We choose a memoryless channel with $\gamma=0.5$. By keeping $K/\log n$ as a constant $15/\log 2$, we change $K$ from $5$ to $300$ and calculate the maximum achievable rate, which is $\eta(n,K)$, through simulations for each pair of $(K,n)$. Since the value of our bound is a function of the ratio $K/\log n$, in this case, it is a constant for all $K$ and is equal to the asymptotic throughput with parameter $15/\log 2$. From Figure~\ref{ex1fig1} we can see that as $K$ approaches infinity, the maximum achievable rate converges to our lower bound (which is also the asymptotic throughput in this case) in a decreasing manner, which validates Theorem~\ref{Theorem1}, Corollary~\ref{MonotoneThMemoryless} and Theorem~\ref{LowerBoundTheorem}.

In this case, we also plot the CSE lower bounds given by Theorem 4. From the figure we can see that our bounds outperforms the CSE lower bounds. CSE bound 1 gradually approaches zero as indicated by Equation~(\ref{cse0}) while our bound is a constant value and asymptotically tight as shown in Equation~(\ref{AsympTight}).
{\color{black}\subsection{Example 2}\label{Example2}
In order to verify Theorem~\ref{LowerBoundTheorem} under the Gilbert-Elliott channel model, we choose the state transition probability $p_{10}=p_{01}=0.4$. It is easy to obtain that $\gamma=0.5$ and $K_0=1$. By keeping $K/\log n$ as a constant $5/\log 2$ and changing $K$ from $5$ to $100$, we plot in Figure~\ref{ex1fig2} both the simulation result of the maximum achievable rate and our lower bound shown in Theorem~\ref{MonotoneTh}. Again, as we can see from the figure, the lower bound becomes tighter as $K$ increases. Eventually the maximum achievable rate will converge to the lower bound as shown in Equation~(\ref{AsympTight}). However, neither of the CSE lower bounds is valid under this system setting. 
}
\subsection{Example 3}\label{Example3}
In this example, we conduct three set of experiments under the memoryless channel assumption with different values of $K$ as a function of $n$, and show that our bound outperforms the CSE bound in all these simulation settings.

\begin{figure*}
\centering
\subfloat[n=K]{\includegraphics[width=2.3in]{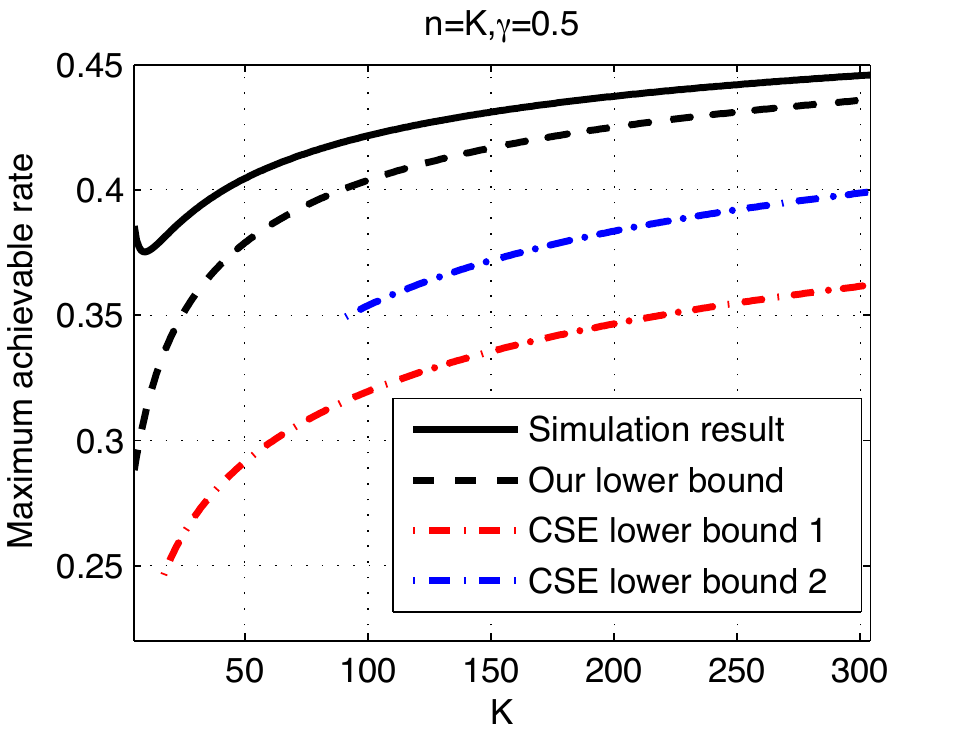} \label{ex2fig1}}
\subfloat[n=10]{\includegraphics[width=2.3in]{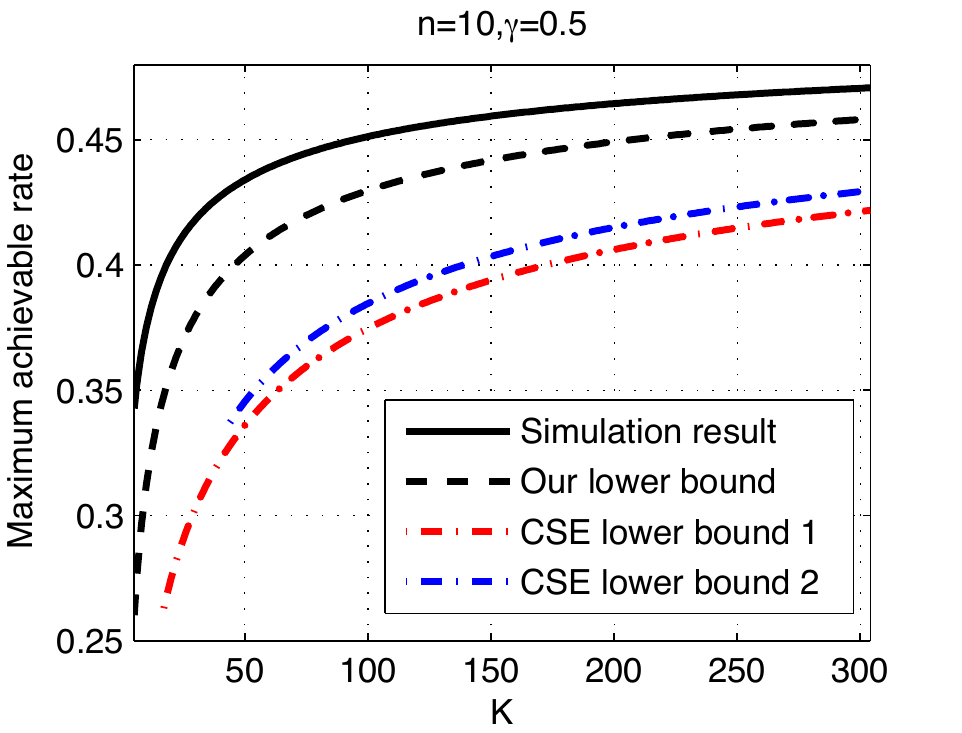} \label{ex2fig2}}
\subfloat[K=80]{\includegraphics[width=2.3in]{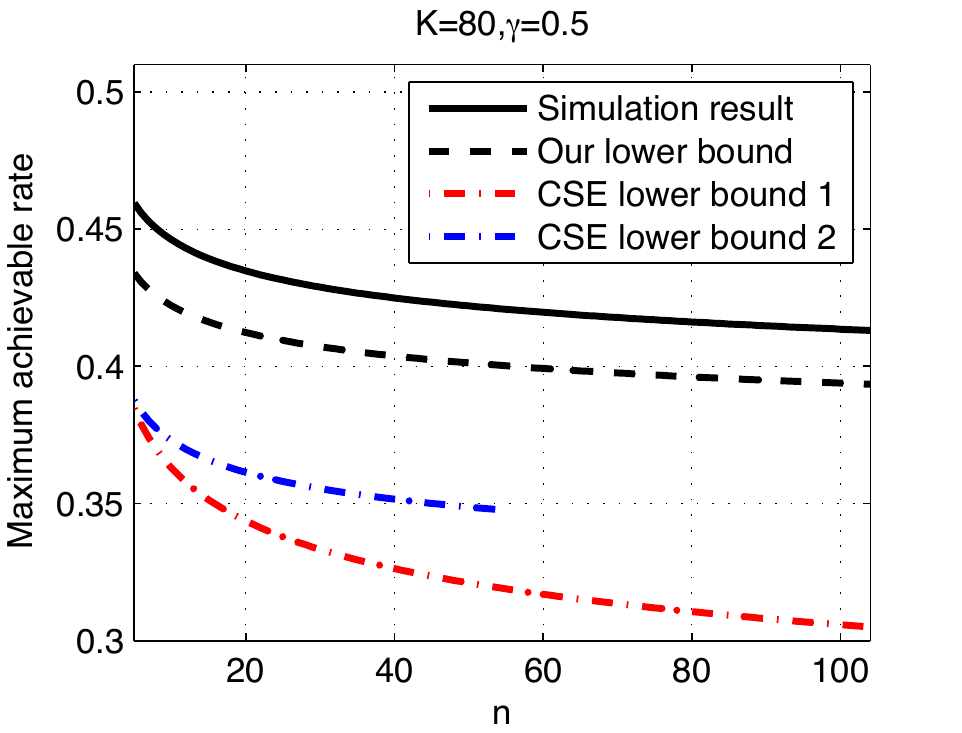} \label{ex2fig3}}\\
\caption{Illustration for example 3}
\label{fig:fourfigures}
\end{figure*}

In the first case, we set the coding block size $K$ to be the same as the network size $n$ and change $n$ from $5$ to $300$. We plot the simulation result of the maximum achievable rate as well as our bound and the CSE bound in Figure~\ref{ex2fig1}, since in this case $K$ scales faster than $\log n$, the achievable rate will approach system capacity $\gamma$ as the network size $n$ grows.

In the second case, we assume that the number of receivers is fixed to be $10$ and we increase coding block size $K$ from $5$ to $300$. The simulations result, together with the two bounds, are plotted in Figure~\ref{ex2fig2}. In this case, the achievable rate will also approach system capacity $\gamma$ as $n$ increases.

In the final case, as shown in Figure~\ref{ex2fig3}, we keep the coding block size to be a constant $80$ and increase the number of receivers from $5$ to $100$. Since $K$ does not increase with $\log n$ at all, the achievable rate will vanish to $0$ as $n$ grows. 

From Figures~\ref{ex1fig1}, \ref{ex2fig1}, \ref{ex2fig2} and \ref{ex2fig3}, we can see that our bound obtained by Theorem~\ref{LowerBoundTheorem} is significantly better than the lower bounds achieved in \cite{Cogill} and \cite{CogillBrooke} in all these four different cases.

\section{Proofs}\label{Proofs}
\subsection{Proof of Theorem~\ref{Theorem1}}\label{p1}
In order to prove Theorem~\ref{Theorem1}, we first need the following lemmas (Lemma~\ref{LemmaTK}, Lemma~\ref{LemmaKn} and Lemma~\ref{previtali}).
\begin{lemma}\label{LemmaTK}
{\rm
For any $\beta\in(0,1)$ and any values of $\mathcal{E}$, we have
	\begin{align}
		&\Pr\left[T(n,K)>\frac{k}{\beta}{\Big|}K=k,\mathcal{E}\right]\notag\\
		=&1-\left(1-e^{-\frac{k}{\beta}\Lambda(\beta,\Pi){\bf1}(\beta<\gamma)+g(\beta,k,\mathcal{E})}\right)^n,\label{lemma1result}
	\end{align}
where
\begin{align}
	g(\beta,k,\mathcal{E})\in\left\{
		\begin{array}{ll}
			o(k)\text{ as }k\to\infty&\text{if }\beta<\gamma\\
			o(1)\text{ as }k\to\infty&\text{if }\beta>\gamma
		\end{array}
	\right.\notag.
\end{align}
}
\end{lemma}
\begin{proof}[Proof of Lemma~\ref{LemmaTK}]
From definition (\ref{TiK}) and (\ref{TnK}), we have, for any $t$,
\begin{align}
\left\{T(n,K)\leq t,\mathcal{E}\right\}=\bigcap_{i=1}^n\left\{T_i(K)\leq t,\mathcal{E}_i\right\}.\notag
\end{align}
Therefore, we have
\begin{align}
&\Pr[T(n,K)> t|K=k,\mathcal{E}]\notag\\
=&1-\Pr[T(n,K)\leq t|K=k,\mathcal{E}]\notag\\
=&1-\prod_{i=1}^n\left(1-\Pr[T_i(K)>t|K=k,\mathcal{E}_i]\right).\label{donotworry1}
\end{align}
Let $t=\frac{k}{\beta}$, from definition \ref{TiK} we can get, for any $1\leq i\leq n$,

\begin{align}
\Pr\left[T_i(K)>\frac{k}{\beta}{\Big|}K=k,\mathcal{E}_i\right]=&\Pr\left[\sum_{j=1}^{k/\beta}X_{ij}<k{\Bigg|}\mathcal{E}_i\right]\notag\\
=&\Pr\left[\frac{\sum_{j=1}^{k/\beta}X_{ij}}{k/\beta}<\beta{\Bigg|}\mathcal{E}_i\right]\notag.
\end{align}
and
\begin{align}
\lim_{k\to\infty}\frac{\log\Pr\left[T_i(K)>\frac{k}{\beta}{\Big|}K=k,\mathcal{E}_i\right]}{k/\beta}=-\Lambda(\beta,\Pi){\bf1}(\beta<\gamma),\label{donotworry2}
\end{align}
with the last equation being a direct application of Theorem 3.1.2 in \cite{Dembozeitouni} (G\"artner-Ellis Theorem for finite state Markov chains). Notice that the right hand side of Equation~(\ref{donotworry2}) is fixed for all possible values of $i$ and $\mathcal{E}_i$ as long as the values of $\beta$ and $\Pi$ are fixed.
Then the proof completes by combining (\ref{donotworry1}) and (\ref{donotworry2}).
\end{proof}
\begin{lemma}\label{LemmaKn}{\rm Assume $k$ is a function of $n$ and denote $k:=k(n)$, and define $f(k,\beta,\mathcal{E}):=e^{\frac{k}{\beta}\Lambda(\beta,\Pi){\bf1}(\beta<\gamma)-g(\beta,k,\mathcal{E})}$, then we have
\begin{enumerate}
\item
For a fixed $\beta\in(0,1)$, if $\lim_{n\to\infty}\frac{n}{f(k(n),\beta,\mathcal{E})}=0$, then
	\begin{align}\label{equal0}
		 \lim_{n\to\infty}\Pr\left[T(n,K)>\frac{k(n)}{\beta}{\Big|}K=k(n),\mathcal{E}\right]=0.
	\end{align}
\item
For a fixed $\beta\in(0,1)$, if $\lim_{n\to\infty}\frac{n}{f(k(n),\beta,\mathcal{E})}=\infty$, then
	\begin{align}\label{equal1}
		 \lim_{n\to\infty}\Pr\left[T(n,K)>\frac{k(n)}{\beta}{\Big|}K=k(n),\mathcal{E}\right]=1.
	\end{align}	
\end{enumerate}
}
\end{lemma}
\begin{proof}[Proof of Lemma~\ref{LemmaKn}]
According to Lemma~\ref{LemmaTK} and the definition of $f(k(n),\beta,\mathcal{E})$, we have
	\begin{align}
		&\Pr\left[T(n,K)>\frac{k(n)}{\beta}{\Big|}K=k(n),\mathcal{E}\right]\notag\\
		=&1-\left(1-\frac{1}{f(k(n),\beta,\mathcal{E})}\right)^n\notag\\
		 =&1-\left[\left(1-\frac{1}{f(k(n),\beta,\mathcal{E})}\right)^{f(k(n),\beta,\mathcal{E})}\right]^{\frac{n}{f(k(n),\beta,\mathcal{E})}}.\notag
	\end{align}
	Since the function $\left(1-\frac{1}{x}\right)^x$ with domain $(1,+\infty)$ is a bounded and strictly increasing function with region $(0,e^{-1})$ and the fact that $f(k,\beta)>1$, we know that if $\lim_{n\to\infty}\frac{n}{f(k(n),\beta,\mathcal{E})}=\infty$, then
	\begin{align}
		 &\liminf_{n\to\infty}\Pr\left[T(n,K)>\frac{k(n)}{\beta}{\Big|}K=k(n),\mathcal{E}\right]\notag\\
		 =&1-\limsup_{n\to\infty}\left[\left(1-\frac{1}{f(k(n),\beta,\mathcal{E})}\right)^{f(k(n),\beta,\mathcal{E})}\right]^{\frac{n}{f(k(n),\beta,\mathcal{E})}}\notag\\
		 \geq&1-\limsup_{n\to\infty}e^{-\frac{n}{f(k(n),\beta,\mathcal{E})}}\notag\\
		=&1,\notag
	\end{align}
	which, together with the fact that $\Pr{\big[}T(n,K)>\frac{k(n)}{\beta}{\big|}K=k(n),\mathcal{E}{\big]}\leq1$, yields Equation~(\ref{equal1}).
	
	If $\lim_{n\to\infty}\frac{n}{f(k(n),\beta,\mathcal{E})}=0$, then $f(k(n),\beta,\mathcal{E})\to\infty$ as $n\to\infty$, which results in
	\begin{align}
	 \lim_{n\to\infty}\left(1-\frac{1}{f(k(n),\beta,\mathcal{E})}\right)^{f(k(n),\beta,\mathcal{E})}=e^{-1}\notag.
	\end{align}
	Then we can obtain
	\begin{align}
		 &\limsup_{n\to\infty}\Pr\left[T(n,K)>\frac{k(n)}{\beta}{\Big|}K=k(n),\mathcal{E}\right]\notag\\
		 =&1-\liminf_{n\to\infty}\left[\left(1-\frac{1}{f(k(n),\beta,\mathcal{E})}\right)^{f(k(n),\beta,\mathcal{E})}\right]^{\frac{n}{f(k(n),\beta,\mathcal{E})}}\notag\\
		 =&1-\notag\\
		 &\liminf_{n\to\infty}\left[\lim_{n\to\infty}\left(1-\frac{1}{f(k(n),\beta,\mathcal{E})}\right)^{f(k(n),\beta,\mathcal{E})}\right]^{\frac{n}{f(k(n),\beta,\mathcal{E})}}\notag\\
		=&1-1=0,\notag
	\end{align}
	which leads to Equation~(\ref{equal0}).
\end{proof}
\begin{lemma}\label{previtali}
Let $\{h_n(x)\}$ be a set of Lebesgue measurable functions defined on $[0,\infty)$ and $h_n(x)$ converges to ${\bf 1}(x<y)$ almost everywhere for some $y>0$. If $h_n(x)$ is a decreasing function of $x$ and have the range $[0,1]$ for any $n\in\mathbb{N}$, then $h_n(x)$ converges globally in measure to ${\bf 1}(x<y)$.
\end{lemma}
\begin{proof}
Choose $\varepsilon>0$.
Since $h_n(x)$ converges to ${\bf 1}(x<y)$ almost everywhere, for any $\delta>0$, we can find $N\in\mathbb{N}$ such that for any $n>N$, we have 
\begin{align}
\left|h_n\left(y-\delta/2\right)-1\right|<\varepsilon\notag\\
\left|h_n\left(y+\delta/2\right)-0\right|<\varepsilon\notag.
\end{align}
Since $0\leq h_n(x)\leq 1$ for any $x\in[0,\infty)$ and $h_n(x)$ is a decreasing function of $x$, we know that, for any $n>M$,
\begin{align}
\begin{array}{ll}
h_n(x)>1-\varepsilon&\text{ }\forall x<y-\delta/2\notag\\
h_n(x)<\varepsilon&\text{ }\forall x>y+\delta/2\notag.
\end{array}
\end{align}
Therefore, for any $n>N$,
\begin{align}
&\nu\left(\left\{|h_n(x)-{\bf1}(x<y)|>\varepsilon\right\}\right)\notag\\
<&\nu([y-\delta/2,y])+\nu([y,y+\delta/2])=\delta,\notag
\end{align}
where $\nu$ is the Lebesgue measure. Since $\varepsilon$ and $\delta$ are arbitrarily chosen, from the above inequality we know that $h_n(x)$ converges globally in measure to ${\bf1}(x<y)$.
\end{proof}

	With Lemma~\ref{LemmaTK},\ref{LemmaKn} and \ref{previtali} established, we now turn to the proof of Theorem~\ref{Theorem1}.
	\begin{proof}[Proof of Theorem~\ref{Theorem1}]
Since $K$ is assumed to be a function of $n$, we denote this function as $k(n)$. According to definition \ref{eta} we have,
{\color{black}
	\begin{align}
		&\lim_{n\to\infty}\left(\eta(n,K)\right)^{-1}\notag\\
		=&\lim_{n\to\infty}\mathbb{E}\left[\frac{\mathbb{E}\left[T(n,K){\big|}\mathcal{E}\right]}{K}\right].\label{etainv}
	\end{align}
	Next, we obtain the value of $\lim_{n\to\infty}\mathbb{E}[T(n,K)|\mathcal{E}]/K$ and show that it is independent of $\mathcal{E}$. Note that
	\begin{align}
		&\lim_{n\to\infty}\mathbb{E}\left[\frac{T(n,K)}{K}{\Big|}\mathcal{E}\right]\notag\\
		=&\lim_{n\to\infty}\int_0^\infty\frac{\Pr\left[T(n,K)>s|K=k(n),\mathcal{E}\right]}{k(n)}ds\notag\\
		=&\lim_{n\to\infty}\int_0^\infty\Pr\left[T(n,K)>k(n)u|K=k(n),\mathcal{E}\right]du.\label{discuss}
	\end{align}
}According to the assumption that $\lim_{n\to\infty}k(n)/\log(n)=c$, we have
	\begin{align}
		&\lim_{n\to\infty}\frac{n}{e^{\frac{k(n)}{\beta}\Lambda(\beta,\Pi){\bf1}(\beta<\gamma)-g(\beta,k,\mathcal{E})}}\notag\\
		=&\left\{
			\begin{array}{ll}	
				0&c>\frac{\beta}{\Lambda(\beta,\Pi){\bf1}(\beta<\gamma)}\\
				\infty&c<\frac{\beta}{\Lambda(\beta,\Pi){\bf1}(\beta<\gamma)}						\end{array}
		\right..\notag
	\end{align}
	Since $\frac{\beta}{\Lambda(\beta,\Pi){\bf1}(\beta<\gamma)}|_{\beta=0}=0$, $\lim_{\beta\to\gamma^-}\frac{\beta}{\Lambda(\beta,\Pi){\bf1}(\beta<\gamma)}=\infty$ and $\frac{\beta}{\Lambda(\beta,\Pi){\bf1}(\beta<\gamma)}$ is a monotone increasing function on the domain $(0,\gamma)$, the equation $c=\frac{\beta}{\Lambda(\beta,\Pi){\bf1}(\beta<\gamma)}$ has only one solution of $\beta$, which we denote as
	\begin{align}
	\beta_c=\sup\left\{\beta{\Big|}c\geq\frac{\beta}{\Lambda(\beta,\Pi)},0\leq\beta<\gamma\right\}.\notag
	\end{align}
	Then, by Lemma~\ref{LemmaKn}, we get
	\begin{align}
	&\lim_{n\to\infty}\Pr\left[T(n,K)>k(n)u{\big|}K=k(n),\mathcal{E}\right]\notag\\
	=&
	\left\{
		\begin{array}{ll}
		1&\text{if }u<\frac{1}{\beta_c}\\
		0&\text{if }u>\frac{1}{\beta_c}
		\end{array}
	\right..\label{sit2}
	\end{align}
	We let $h_n(u)\triangleq\Pr{\big[}T(n,K)>k(n)u{\big|}K=k(n),\mathcal{E}{\big]}$. Equation~(\ref{sit2}) implies that $h_n(u)$ converges to ${\bf 1}(u<1/\beta_c)$ pointwisely. Since $h_n(u)$ is a decreasing function of $u$ and has the range $[0,1]$ for all $n$, by Lemma~\ref{previtali} we know that $h_n(u)$ globally converges in measure to ${\bf 1}(u<1/\beta_c)$. We also know that the set of function $\{h_n(u)\}$ is uniformly bounded. Then we can apply Vitali convergence {\color{black}theorem} to Equation~(\ref{discuss}) to exchange the limit and integral and obtain
	\begin{align}
	&\lim_{n\to\infty}\mathbb{E}\left[\frac{T(n,K)}{K}{\Big|}\mathcal{E}\right]\notag\\
	=&\int_0^\infty\lim_{n\to\infty}\Pr\left[T(n,K)>k(n)u{\big|}K=k(n),\mathcal{E}\right]du=\frac{1}{\beta_c}.\label{independentofe}
	\end{align}
	{\color{black}Note that the above result is independent of the choice of the initial state $\mathcal{E}$. Since the cardinality of the state space of $\mathcal{E}$ is finite for a finite value of $n$, we can exchange the limit and expectation in Equation~(\ref{etainv}), which, after combining with the above equation, completes the proof.}
\end{proof}

{\color{black}
\subsection{Proof of Theorem~\ref{MonotoneTh}.}
Let us define two random variables $T^{(0)}$ and $T^{(1)}$ under the Gilbert-Elliott channels as
	\begin{align}
		T^{(0)}=\min_{m}{\Bigg\{}m{\bigg|}\sum_{j=1}^mX_{1j}\geq1,X_{10}=0{\Bigg\}}\label{Ts0},\\
		T^{(1)}=\min_{m}{\Bigg\{}m{\bigg|}\sum_{j=1}^mX_{1j}\geq1,X_{10}=1{\Bigg\}}\label{Ts1}.
	\end{align}
	In order to prove Theorem~\ref{MonotoneTh}, we first need the following lemma.
\begin{lemma}\label{stochasticordering}
	Let ${\big\{}T^{(1)}_d{\big\}}_{d\in\mathbb{N}}$ be i.i.d.~random variables with the same distribution as $T^{(1)}$, then we have
	\begin{align}
	\sum_{d=1}^{K_0+1}T^{(1)}_d\succeq T^{(0)},\notag
	\end{align}
	meaning that $\sum_{d=1}^{K_0}T^{(1)}_d$ is stochastically greater than or equal to $T^{(0)}$, where
	\begin{align}
	K_0=\min{\Bigg\{}m\geq0{\Bigg|}\sum_{d=0}^m(1-p_{10})^{d}p_{10}+p_{01}\geq1{\Bigg\}}\label{k0}.
	\end{align}
\end{lemma}
\begin{proof}[Proof of Lemma~\ref{stochasticordering}]
First observe that $\sum_{d=0}^\infty(1-p_{10})^{d}p_{10}=1$, which makes sure that $K_0$ in Equation~(\ref{k0}) is well defined.

Then according to the definition of $T^{(0)}$ and $T^{(1)}$ in Equations~(\ref{Ts0}) and (\ref{Ts1}), we have, for any integer $1\leq t\leq K_0+1$,
	\begin{align}
		\Pr{\Bigg[}\sum_{d=1}^{K_0+1}T_d^{(1)}>t{\Bigg]}=1\geq\Pr\left[T^{(0)}>t\right]\label{orderingeq},
	\end{align}
and for any integer $t>K_0+1$,
	\begin{align}
		&\Pr{\Bigg[}\sum_{d=1}^{K_0+1}T_d^{(1)}>t{\Bigg]}\notag\\
		\geq&\sum_{d=1}^{K_0+1}\Pr\left[T^{(1)}>1\right]\Pr\left[T^{(1)}=1\right]^{d-1}\Pr\left[T^{(0)}>t-d\right]\notag\\
		\geq&\sum_{d=1}^{K_0+1}\Pr\left[T^{(1)}>1\right]\Pr\left[T^{(1)}=1\right]^{d-1}\Pr\left[T^{(0)}>t-1\right]\notag\\
		=&\sum_{d=1}^{K_0+1}p_{10}(1-p_{10})^{d-1}(1-p_{01})^{t-1}\notag\\
		\geq&(1-p_{01})(1-p_{01})^{t-1}=(1-p_{01})^t=\Pr\left[T^{(0)}>t\right]\notag,
	\end{align}
	with the last inequality followed by the definition of $K_0$ in Equation~(\ref{k0}). The above equation, together with Equation~(\ref{orderingeq}), completes the proof.
\end{proof}

With Lemma~\ref{stochasticordering} established, we now turn to the proof of Theorem~\ref{MonotoneTh}.
\begin{proof}[Proof of Theorem~\ref{MonotoneTh}]\label{MonotoneProof}
Under the Gilbert-Elliott channel assumption as illustrated in Figure~\ref{GilbertElliottIll}, let $T(n,K,\mathcal{E})$ be defined as $T(n,K)$ with initial status $\mathcal{E}$. Then according to Definitions~\ref{TiK} and \ref{TnK} and Equations~(\ref{Ts0}) and (\ref{Ts1}) we can express $T(n,K,\mathcal{E})$ as
	\begin{align}
	T(n,K,\mathcal{E})&=\max_{1\leq i\leq n}{\Bigg\{}T_{i1}^{(\mathcal{E}_i)}+\sum_{j=2}^KT_{ij}^{(1)}{\Bigg\}}\notag,
\end{align}
where $\{T_{ij}^{(0)}\}_{i,j\in\mathbb{N}}$ are i.i.d.~random variables with the same distribution as $T^{(0)}$, and $\{T_{ij}^{(1)}\}_{i,j\in\mathbb{N}}$ are i.i.d.~random variables with the same distribution as $T^{(1)}$. Similarly we can express $T(n^\alpha,\alpha K,{\bf 1}_{n^\alpha})$ and $T(n^\alpha,\alpha K,{\bf 1}_{n^\alpha})$ as
	\begin{align}
	T(n,K,{\bf 1}_{n})&=\max_{1\leq i\leq n}{\Bigg\{}\sum_{j=1}^KT_{ij}^{(1)}{\Bigg\}}\label{tt1},\\
	T(n^\alpha,\alpha K,{\bf 1}_{n^\alpha})&=\max_{1\leq i\leq n^\alpha}{\Bigg\{}\sum_{j=1}^{\alpha K}T_{ij}^{(1)}{\Bigg\}}\label{tt2}.
	\end{align}
	First, by Lemma~\ref{stochasticordering} we know that, for any $1\leq i\leq n$ and any initial status $\mathcal{E}$,
	\begin{align}
		\sum_{j=1}^{K_0+1}T_{ij}^{(1)}+\sum_{j=K_0+2}^{K+K_0}T_{ij}^{(1)}&\succeq T_{i1}^{(\mathcal{E}_i)}+\sum_{j=2}^KT_{ij}^{(1)}\notag,
	\end{align}
	implying that
	\begin{align}
		\max_{1\leq i\leq n}{\Bigg\{}\sum_{j=1}^{K+K_0}T_{ij}^{(1)}{\Bigg\}}&\succeq \max_{1\leq i\leq n}{\Bigg\{}T_{i1}^{(\mathcal{E}_i)}+\sum_{j=2}^KT_{ij}^{(1)}{\Bigg\}}\notag,
	\end{align}
	which yields
	\begin{align}
	\mathbb{E}[T(n,K+K_0,{\bf1}_n)]\geq\mathbb{E}[T(n,K,\mathcal{E})].\label{orderingresult1}
	\end{align}
	Next, we will show that 
	\begin{align}
		\mathbb{E}[T(n^{\alpha},\alpha K,{\bf1}_{n^\alpha})]\geq\alpha\mathbb{E}[T(n,K,{\bf 1}_{n})]\notag.
	\end{align}
Let us deonte\begin{align}S_{i}^{r}=\sum_{j=1+(r-1)K}^{rK}T_{ij}^{(1)}.\notag\end{align}Then we know that $\{S_{i}^{r}\}_{i\in\mathbb{N},r\in\mathbb{N}}$ are  i.i.d.~random variables. Equation~(\ref{tt1}) and (\ref{tt2}) can be rewritten as
\begin{align}
T(n,K,{\bf1}_n)&=\max_{1\leq i\leq n}S_i^1\label{tnks}\\
T(n^\alpha,\alpha K,{\bf1}_{n^\alpha})&=\max_{1\leq i\leq n^\alpha}\sum_{r=1}^\alpha S_{i}^{r}\label{tnksa}.
\end{align}
Instead of viewing Equation~(\ref{tnksa}) as a 1-dimensional maximization over $n^\alpha$ points, we can think of it as an $\alpha$-dimensional maximization over $n^\alpha$ points where we can choose a coordinate from $1$ to $n$ on each dimension and therefore can further rewrite Equation~(\ref{tnksa}) as
\begin{align}
T(n^\alpha,\alpha K,{\bf1}_{n^\alpha})=\max_{1\leq i_1\leq n}\max_{1\leq i_2\leq n}\ldots\max_{1\leq i_\alpha\leq n}\sum_{r=1}^\alpha S^r_{(i_1,i_2,\ldots,i_\alpha)}\label{tnksa2},
\end{align}
where 
\begin{align}
S^r_{(i_1,i_2,\ldots,i_\alpha)}=S^r_{\sum_{u=1}^\alpha n^{u-1}(i_u-1)+1}\notag
\end{align}
and $i_u$ can be viewed as the coordinate in the $u^{\text{th}}$ dimension.

Next, we use Equation~(\ref{tnksa2}) to build a lower bound on the expection of $T(n^\alpha,\alpha K,{\bf1}_{n^\alpha})$.

For fixed values of $i_2,i_3,\ldots,i_\alpha$, let us find a $i_1^*$ such that
\begin{align}
i_1^*\left(i_2,\ldots,i_\alpha\right)=\arg\max_{1\leq i_1\leq n}S^1_{(i_1,i_2,\ldots,i_\alpha)}\label{istardef},
\end{align}
which we denote as $i_1^*$ for short.
Then according to Equation~(\ref{tnksa2}), we can find a lower bound for $\mathbb{E}[T(n^\alpha,\alpha K,{\bf1}_{n^\alpha})]$ by choosing $i_1=i_1^*\left(i_2,\ldots,i_\alpha\right)$ for all possible values of $i_2,i_3,\ldots,i_\alpha$, which is
\begin{align}
&\mathbb{E}\left[T(n^\alpha,\alpha K),{\bf1}_{n^\alpha}\right]\notag\\
=&\mathbb{E}\left[\max_{1\leq i_1\leq n}\max_{1\leq i_2\leq n}\ldots\max_{1\leq i_\alpha\leq n}\sum_{r=1}^\alpha S_{(i_1,i_2,\ldots,i_\alpha)}^r\right]\notag\\
\overset{(a)}{\geq}&\mathbb{E}\left[\max_{1\leq i_2\leq n}\ldots\max_{1\leq i_\alpha\leq n}\sum_{r=1}^\alpha S_{(i_1^*,i_2,\ldots,i_\alpha)}^r\right]\notag\\
=&\mathbb{E}\left[\max_{1\leq i_2\leq n}\ldots\max_{1\leq i_\alpha\leq n}\left(\sum_{r=2}^\alpha S_{(i_1^*,i_2,\ldots,i_\alpha)}^r+S^1_{(i_1^*,i_2,\ldots,i_\alpha)}\right)\right].\label{tnki1star}
\end{align}
Since the choice of $i_1^*$ is only sub-optimal, the inequality (a) in Equation~(\ref{tnki1star}) should be strict inequality. Notice that according to Equation~(\ref{istardef}), for any values of $i_2,i_3,\ldots,i_\alpha$, we have
\begin{align}
S^1_{(i_1^*,i_2,\ldots,i_\alpha)}=\max_{1\leq i_1\leq n} S^1_{(i_1,i_2,\ldots,i_\alpha)},\notag
\end{align}
which, combining Equation~(\ref{tnks}) and the fact that $\{S_i^r\}$ are i.i.d.~random variables, yields
\begin{align}
\mathbb{E}\left[S^1_{(i_1^*,i_2,\ldots,i_\alpha)}\right]=&\mathbb{E}\left[\max_{1\leq i_1\leq n}S_{(i_1,i_2,\ldots,i_\alpha)}^1\right]\notag\\
=&\mathbb{E}\left[\max_{1\leq i\leq n}S_{i}^1\right]\notag\\
=&\mathbb{E}\left[T(n,K,{\bf1}_n)\right].\label{etnk1}
\end{align}
As a second step, for any values of $i_3,i_4,\ldots,i_\alpha$, let us define $i_2^*$ as
\begin{align}
i_2^*(i_1^*,i_3,\ldots,i_\alpha)=\arg\max_{1\leq i_2\leq n}S^2_{(i_1^*,i_2,\ldots,i_\alpha)}.\notag
\end{align}
Then similarly as Equation~(\ref{tnki1star}), by fixing $i_2$ to be $i_2^*$, we can obtain
\begin{align}
\mathbb{E}{\big[}&T(n^\alpha,\alpha K,{\bf1}_{n^\alpha}){\big]}\notag\\
>\mathbb{E}{\Bigg[}&\max_{1\leq i_3\leq n}\ldots\max_{1\leq i_\alpha\leq n}\notag\\
&\left(\sum_{r=3}^\alpha S_{(i_1^*,i_2^*,\ldots,i_\alpha)}^r+S^1_{(i_1^*,i_2^*,\ldots,i_\alpha)}+S^2_{(i_1^*,i_2^*,\ldots,i_\alpha)}\right){\Bigg]}.\notag
\end{align}
Also, for any values of $i_3,i_4\ldots,i_\alpha$, we have
\begin{align}
\mathbb{E}\left[S^2_{(i_1^*,i_2^*,\ldots,i_\alpha)}\right]=&\mathbb{E}\left[\max_{1\leq i_2\leq n}S_{(i_1^*,i_2,\ldots,i_\alpha)}^1\right]\notag\\
=&\mathbb{E}\left[T(n,K,{\bf1}_n)\right]\label{etnk2}.
\end{align}
By defining $i_3^*,\ldots,i_\alpha^*$ in a similar way 
\begin{align}
&i_u^*(i_1^*,\ldots,i_{u-1}^*,i_{u+1},\ldots,i_\alpha)\notag\\
=&\arg\max_{1\leq i_u\leq n}S^u_{(i_1^*,\ldots,i_{u-1}^*,i_u,\ldots,i_\alpha)}\notag
\end{align}
and iterating the above step, we can get
\begin{align}
&\mathbb{E}{\big[}T(n^\alpha,\alpha K,{\bf1}_{n^\alpha}){\big]}\notag\\
>&\mathbb{E}\left[S^1_{(i_1^*,i_2^*,\ldots,i_\alpha^*)}+S^2_{(i_1^*,i_2^*,\ldots,i_\alpha^*)}+\ldots+S^\alpha_{(i_1^*,i_2^*,\ldots,i_\alpha^*)}\right]\notag\\
\overset{(b)}{=}&\sum_{r=1}^\alpha\mathbb{E}\left[S^r_{(i_1^*,i_2^*,\ldots,i_\alpha^*)}\right]\notag\\
\overset{(c)}{=}&\alpha\mathbb{E}\left[T(n,K,{\bf1}_n)\right]\label{alpharel}.
\end{align}
Equation (b) follows from the fact that $\{S^r_{(i_1^*,i_2^*,\ldots,i_\alpha^*)}\}_{1\leq r\leq \alpha}$ are independent random variables and equation (c) follows from Equations~(\ref{etnk1}), (\ref{etnk2}), and iterative steps. 
By combining Equations~(\ref{Thput}), (\ref{orderingresult1}), and (\ref{alpharel}), we have,
\begin{align}
\eta(n,K)&=\frac{K}{\mathbb{E}[\mathbb{E}[T(n,K,\mathcal{E})]]}\notag\\
&\geq\frac{K}{\mathbb{E}[T(n,K+K_0,{\bf1}_{n})]}\notag\\
&>\frac{\alpha K}{\mathbb{E}\left[T(n^\alpha,\alpha(K+K_0),{\bf1}_{n^\alpha})\right]}\notag\\
&=\frac{K}{K+K_0}\frac{\alpha (K+K_0)}{\mathbb{E}\left[T(n^\alpha,\alpha(K+K_0),{\bf1}_{n^\alpha})\right]}\notag,
\end{align}
which completes the proof.
\end{proof}}
\section{Conclusion}\label{Conclusion}
In this paper, we characterize the throughput of a broadcast network using rateless codes. The broadcast channels are modeled by Markov modulated packet erasure channels, where the packet can either be erased or successfully received and for each receiver the current channel state distribution depends on the channel states in previous $l$ packet transmissions. 

We first characterize the asymptotic throughput of the system when $n$ approaches infinity for any values of the coding block size $K$ as a function of number of receivers $n$ in an explicit form. We show that as long as $K$ scales at least as fast as $\log n$, we can achieve a non-zero asymptotic throughput. {\color{black}Under the more restrictive Gilbert-Elliott channel model ($l=1$), we study the case when $K$ and $n$ are finite. For any $K$ and $n$, we find a lower bound on the throughput in terms of the transmission time of a system with larger $K$ and $n$. As a special case when channels are memoryless, this result shows that, by keeping the ratio $K/\log n$ to be a constant, the system throughput will converge to the asymptotic throughput in a decreasing manner as $n$ grows. By the help of these results, under the Gilbert-Elliott channel model, we are able to give a lower bound on the maximum achievable throughput (maximum achievable rate), which is a function of $K$, $n$ and state transition probabilities $p_{01}$ and $p_{10}$. In contrast to the state-of-the-art, we analytically show that our bound is asymptotically tight when $K/\log n$ is fixed as $n$ approaches infinity. Further, through numerical evaluations, we show that our bound is significantly better than existing results.}

\section{Acknowledgments}
{\color{black}
The authors would like to thank Dr.~Yin Sun for the valuable discussion that inspired the proof of Theorem~\ref{MonotoneTh} and Swapna B. T. for her helpful comment on the definition of throughput. 

This work was supported in part by NSF grants CNS-0905408, CNS-1012700, from the Army Research Office MURI grant W911NF-08-1-0238, and an HP IRP award.
}
\bibliography{IEEEabrv,yangbib} 

\end{document}